\def\year{2020}\relax
\documentclass[letterpaper]{article} % DO NOT CHANGE THIS

\usepackage{aaai20}%.sty % DO NOT CHANGE THIS
\usepackage{courier} % DO NOT CHANGE THIS
\usepackage{graphicx} % DO NOT CHANGE THIS
\usepackage{helvet} % DO NOT CHANGE THIS
\usepackage{times} % DO NOT CHANGE THIS

\usepackage[hyphens]{url} % DO NOT CHANGE THIS

\title{ADDMC: Weighted Model Counting with Algebraic Decision Diagrams}
\author{
  \Large \textbf{Jeffrey M. Dudek, Vu H. N. Phan, Moshe Y. Vardi}
  \\
  Rice University \\ % If you have multiple authors and multiple affiliations, use superscripts in text and roman font to identify them. For example, Sunil Issar,\textsuperscript{\rm 2} J. Scott Penberthy\textsuperscript{\rm 3} George Ferguson,\textsuperscript{\rm 4} Hans Guesgen\textsuperscript{\rm 5}. Note that the comma should be placed BEFORE the superscript for optimum readability.
  6100 Main Street \\
  Houston, Texas 77005 \\
  \{jmd11,vhp1,vardi\}@rice.edu % email address must be in roman text type, not monospace or sans serif
}

\usepackage{amsmath}
\usepackage{amssymb}
\usepackage{amsthm}

\usepackage{multirow}

\usepackage{enumitem}
\setlist[enumerate]{label*=\arabic*.}

\theoremstyle{definition} % non-italic

\usepackage[linesnumbered,ruled]{algorithm2e}
\SetCommentSty{}
\SetKwIF{If}{ElseIf}{Else}{if}{}{else if}{else}{}
\SetKwFor{For}{for}{}{}

\newcommand{\pars}[1]{\left( #1 \right)}

\newcommand{\size}[1]{\left| #1 \right|}

\newcommand{\set}[1]{\left\{ #1 \right\}}
\newcommand{\sete}{\varnothing}

\newcommand{\Set}[1]{\mathbb{#1}}
\newcommand{\R}{\Set{R}}
\newcommand{\Z}{\Set{Z}}
\newcommand{\N}{\Z^+}

\newcommand{\wrt}{w.r.t.}
\newcommand{\eg}{e.g.}
\newcommand{\ie}{i.e.}

\newcommand{\textdef}[1]{\textit{#1}}

\newcommand{\heuristic}[1]{\textbf{#1}}

\newcommand{\Mono}{\heuristic{Mono}}

\newcommand{\Be}{\heuristic{BE}}
\newcommand{\Bm}{\heuristic{BM}}

\newcommand{\List}{\heuristic{List}}
\newcommand{\Tree}{\heuristic{Tree}}

\newcommand{\Random}{\heuristic{Random}}
\newcommand{\Mcs}{\heuristic{MCS}}
\newcommand{\Lexp}{\heuristic{LexP}}
\newcommand{\Lexm}{\heuristic{LexM}}
\newcommand{\Invmcs}{\heuristic{InvMCS}}
\newcommand{\Invlexp}{\heuristic{InvLexP}}
\newcommand{\Invlexm}{\heuristic{InvLexM}}

\newcommand{\heuristicConfigCount}{245} % |{Mono, BE-List, BE-Tree, BM-List, BM-Tree}| * |{Random, MCS, LexP, LexM, InvMCS, InvLexP, InvLexM}|^2

\newcommand{\urlBenchmarksBayes}{\url{https://www.cs.rochester.edu/u/kautz/Cachet/}}
\newcommand{\urlBenchmarksPseudoweighted}{\url{http://www.cril.univ-artois.fr/KC/benchmarks.html}}

\newcommand{\class}[1]{\textbf{#1}}
\newcommand{\family}[1]{\textit{#1}}

\newcommand{\classBayes}{\class{Bayes}}

\newcommand{\classPseudoweighted}{\class{Non-Bayes}}
\newcommand{\famBmc}{\family{Bounded Model Checking (BMC)}}
\newcommand{\famCircuit}{\family{Circuit}}
\newcommand{\famConfig}{\family{Configuration}}
\newcommand{\famHandmade}{\family{Handmade}}
\newcommand{\famPlanning}{\family{Planning}}
\newcommand{\famQif}{\family{Quantitative Information Flow (QIF)}}
\newcommand{\famRandom}{\family{Random}}
\newcommand{\famSchedule}{\family{Scheduling}}

\newcommand{\benchmarkCountBayes}{1091}
\newcommand{\benchmarkCountPseudoweighted}{823}
\newcommand{\benchmarkCountAltogether}{1914}

\newcommand{\benchmarkCountUnique}{124}
\newcommand{\benchmarkCountFastest}{763}
\newcommand{\benchmarkCountFastestNonUnique}{639}

\newcommand{\benchmarkCountSolved}{1404}

\newcommand{\benchmarkCountMavc}{1906}

\newcommand{\timeoutAllHeuristics}{10}

\newcommand{\timeoutCounters}{1000}
\newcommand{\timeoutMavc}{10000}

\newcommand{\minMavc}{4}
\newcommand{\maxSolvedMavc}{246}

\newcommand{\tool}[1]{\texttt{#1}}

\newcommand{\addmc}{\tool{ADDMC}}
\newcommand{\cachet}{\tool{Cachet}}
\newcommand{\ctd}{\tool{c2d}}
\newcommand{\cudd}{\tool{CUDD}}
\newcommand{\df}{\tool{d4}}
\newcommand{\minictd}{\tool{miniC2D}}
\newcommand{\reasoner}{\tool{d-DNNF-reasoner}}

\newcommand{\sylvan}{\tool{Sylvan}}

\newcommand{\besto}{\textit{Best1}}
\newcommand{\bestt}{\textit{Best2}}

\newcommand{\vbs}{\tool{VBS1}}
\newcommand{\vbss}{\tool{VBS0}}

\newcommand{\weightedCounters}{\cachet, \ctd, \df, and \minictd}

\usepackage[hidelinks]{hyperref}

\newtheorem{definition}{Definition}
\newtheorem{lemma}{Lemma}
\newtheorem{theorem}{Theorem}

\newcommand{\mult}{\cdot}

\newcommand{\func}[1]{\texttt{#1}}
\newcommand{\get}[1]{\func{get-#1}}

\newcommand{\getClusterVarOrder}{\get{cluster-var-order}}
\newcommand{\getDiagramVarOrder}{\get{diagram-var-order}}
\newcommand{\chooseCluster}{\func{choose-cluster}}

\newcommand{\vars}{\texttt{Vars}}

\begin{document}

%%%%%%%%%%%%%%%%%%%%%%%%%%%%%%%%%%%%%%%%%%%%%%%%%%%%%%%%%%%%%%%%%%%%%%%%%%%%%%%%

\maketitle

%%%%%%%%%%%%%%%%%%%%%%%%%%%%%%%%%%%%%%%%%%%%%%%%%%%%%%%%%%%%%%%%%%%%%%%%%%%%%%%%

\begin{abstract}
We present an algorithm to compute exact literal-weighted model counts of Boolean formulas in Conjunctive Normal Form.
Our algorithm employs dynamic programming and uses Algebraic Decision Diagrams as the main data structure.
We implement this technique in \addmc, a new model counter.
We empirically evaluate various heuristics that can be used with \addmc.
We then compare \addmc{} to four state-of-the-art weighted model counters (\weightedCounters) on \benchmarkCountAltogether{} standard model counting benchmarks and show that \addmc{} significantly improves the virtual best solver.
\end{abstract}

%%%%%%%%%%%%%%%%%%%%%%%%%%%%%%%%%%%%%%%%%%%%%%%%%%%%%%%%%%%%%%%%%%%%%%%%%%%%%%%%
%% file start 1-introduction.tex
%%%%%%%%%%%%%%%%%%%%%%%%%%%%%%%%%%%%%%%%%%%%%%%%%%%%%%%%%%%%%%%%%%%%%%%%%%%%%%%%

\section{Introduction}

\textdef{Model counting} is a fundamental problem in artificial intelligence, with applications in machine learning, probabilistic reasoning, and verification \cite{DH07,BHM09,NRJKVMS07}.
Given an input set of constraints, with the focus in this paper on Boolean constraints, the model counting problem is to count the number of satisfying assignments.
Although this problem is \#P-Complete \cite{valiant79}, a variety of tools exist that can handle industrial sets of constraints, cf. \cite{SBBK04,OD15}.

Dynamic programming is a powerful technique that has been applied across computer science \cite{howard66}, including to model counting \cite{BDP09,SS10}.
The key idea is to solve a large problem by solving a sequence of smaller subproblems and then incrementally combining these solutions into the final result.
Dynamic programming provides a natural framework to solve a variety of problems defined on sets of constraints: subproblems can be formed by partitioning the constraints into sets, called \textdef{clusters}.
This framework has also been instantiated into algorithms for database-query optimization \cite{MPPV04} and SAT-solving \cite{US94,AV01,PV04}.
Techniques for local computation can also be seen as a variant of this framework, \eg, in theorem proving \cite{wilson1999logical} or probabilistic inference \cite{shenoy2008axioms}.

In this work, we study two algorithms that follow this dynamic-programming framework and can be adapted for model counting: \textdef{bucket elimination} \cite{dechter99} and \textdef{Bouquet's Method} \cite{bouquet99}.
Bucket elimination aims to minimize the amount of information needed to be carried between subproblems.
When this information must be stored in an uncompressed table, bucket elimination will, with some carefully chosen sequence of clusters, require the minimum possible amount of intermediate data, as governed by the treewidth of the input formula \cite{BDP09}.
Intermediate data, however, need not be stored uncompressed.
Several works have shown that using compact representations of intermediate data can dramatically improve bucket elimination for Bayesian inference \cite{PZ03,sanner2005affine,CD07}.
Moreover, it has been observed that using compact representations --- in particular, {Binary Decision Diagrams (BDDs)} --- can allow Bouquet's Method to outperform bucket elimination for SAT-solving \cite{PV04}.
% This approach can be lifted for model counting, and a similar observation can be made.
Compact representations are therefore promising to improve existing dynamic-programming-based algorithms for model counting \cite{BDP09,SS10}.

In particular, we consider the use of \textdef{Algebraic Decision Diagrams (ADDs)} \cite{BFGHMPS97} for model counting in a dynamic-programming framework.
An ADD is a compact representation of a real-valued function as a directed acyclic graph.
For functions with logical structure, an ADD representation can be exponentially smaller than the explicit representation.
ADDs have been successfully used as part of dynamic-programming frameworks for Bayesian inference \cite{CD07,gogate2012approximation} and stochastic planning \cite{HSHB99}.
Although ADDs have been used for model counting outside of a dynamic-programming framework \cite{fargier2014knowledge}, no prior work uses ADDs for model counting as part of a dynamic-programming framework.

The construction and resultant size of an ADD depend heavily on the choice of an order on the variables of the ADD, called a \textdef{diagram variable order}.
Some variable orders may produce ADDs that are exponentially smaller than others for the same real-valued function.
A variety of techniques exist in prior work to heuristically find diagram variable orders \cite{tarjan1984simple,koster2001treewidth}.
In addition to the diagram variable order, both bucket elimination and Bouquet's Method require another order on the variables to build and arrange the clusters of input constraints; we call this a \textdef{cluster variable order}.
We show that the choice of heuristics to find cluster variable orders has a significant impact on the runtime performance of both bucket elimination and Bouquet's Method.

The primary contribution of this work is a dynamic-programming framework for weighted model counting that utilizes ADDs as a compact data structure.
In particular:
\begin{enumerate}
  \item We lift the BDD-based approach for Boolean satisfiability of \cite{PV04} to an ADD-based approach for weighted model counting.
  \item We implement this algorithm using ADDs and a variety of existing heuristics to produce \addmc{}, a new weighted model counter.
  \item We perform an experimental comparison of these heuristic techniques in the context of weighted model counting.
  \item We perform an experimental comparison of \addmc{} to four state-of-the-art weighted model counters (\weightedCounters) and show that \addmc{} improves the virtual best solver on 763 of \benchmarkCountAltogether{} benchmarks.
\end{enumerate}

% In Section \ref{sec_theory}, we outline the theoretical foundation for performing weighted model counting with ADDs.
% In Section \ref{sec_algorithm}, we present an algorithm for performing weighted model counting through dynamic-programming techniques, and discuss a variety of existing heuristics that can be used in the algorithm.
% In Section \ref{sec_experiments}, we experimentally evaluate \addmc{} heuristics and compare \addmc{} to state-of-the-art weighted model counters (\weightedCounters).
% % In Section \ref{sec_experiments_heuristics}, we compare the performance of various heuristics in \addmc{} and demonstrate that Bouquet's Method is competitive with bucket elimination.
% % In Section \ref{sec_experiments_counters}, we compare \addmc{} against state-of-the-art model counters \weightedCounters{} on two benchmark classes: \classBayes{} and \classPseudoweighted.
% % Through further analysis, we observe that \addmc{} is less competitive on the \classPseudoweighted{} benchmarks but is a dramatic improvement over other counters on the \classBayes{} benchmarks.
% Finally, we conclude in Section \ref{sec_discussion}.

%%%%%%%%%%%%%%%%%%%%%%%%%%%%%%%%%%%%%%%%%%%%%%%%%%%%%%%%%%%%%%%%%%%%%%%%%%%%%%%%
%% file end 1-introduction.tex
%%%%%%%%%%%%%%%%%%%%%%%%%%%%%%%%%%%%%%%%%%%%%%%%%%%%%%%%%%%%%%%%%%%%%%%%%%%%%%%%

%%%%%%%%%%%%%%%%%%%%%%%%%%%%%%%%%%%%%%%%%%%%%%%%%%%%%%%%%%%%%%%%%%%%%%%%%%%%%%%%
%% file start 2-preliminaries.tex
%%%%%%%%%%%%%%%%%%%%%%%%%%%%%%%%%%%%%%%%%%%%%%%%%%%%%%%%%%%%%%%%%%%%%%%%%%%%%%%%

\section{Preliminaries}

In this section, we introduce weighted model counting, the central problem of this work, and Algebraic Decision Diagrams, the primary data structure we use to solve weighted model counting.

%%%%%%%%%%%%%%%%%%%%%%%%%%%%%%%%%%%%%%%%%%%%%%%%%%%%%%%%%%%%%%%%%%%%%%%%%%%%%%%%
\subsection{Weighted Model Counting}

The central problem of this work is to compute the weighted model count of a Boolean formula, which we now define.

\begin{definition}
\label{def_wmc}
  Let $\varphi: 2^X \rightarrow \set{0, 1}$ be a Boolean function over a set $X$ of variables, and let $W: 2^X \rightarrow \mathbb{R}$ be an arbitrary function.
  The \textdef{weighted model count} of $\varphi$ \wrt{} $W$ is
  $$W(\varphi) = \sum_{\tau \in 2^X} \varphi(\tau) \mult W(\tau).$$
\end{definition}

The function $W: 2^X \rightarrow \mathbb{R}$ is called a \textdef{weight function}.
In this work, we focus on so-called \textdef{literal-weight functions}, where the weight of a model can be expressed as the product of weights associated with all satisfied literals.
That is, where the weight function $W$ can be expressed, for all $\tau \in 2^X$, as
$$W(\tau) = \prod_{x \in \tau} W^+(x) \mult \prod_{x \in X \setminus \tau} W^-(x)$$
for some functions $W^+(x), W^-(x): X \rightarrow \mathbb{R}$.
One can interpret these literal-weight functions $W$ as assigning a real-valued weight to each literal: $W^+(x)$ to $x$ and $W^-(x)$ to $\neg x$.
It is common to restrict attention to weight functions whose range is $\R$ or just the interval $[0, 1]$.

When the formula $\varphi$ is given in \textdef{Conjunctive Normal Form (CNF)}, computing the literal-weighted model count is \#P-Complete \cite{valiant79}.
Several algorithms and tools for weighted model counting directly reason about the CNF representation.
For example, \cachet{} uses DPLL search combined with component caching and clause learning to perform weighted model counting \cite{SBBK04}.

If $\varphi$ is given in a compact representation --- \eg, as a Binary Decision Diagram (BDD) \cite{bryant86} or as a Sentential Decision Diagram (SDD) \cite{darwiche11} --- computing the literal-weighted model count can be done in time polynomial in the size of the representation.
One recent tool for weighted model counting that exploits this is \minictd{}, which compiles the input CNF formula into an SDD and then performs a polynomial-time count on the SDD \cite{OD15}.
Although usually more succinct than a truth table, such compact representations may still be exponential in the size of the CNF formula in the worst case \cite{bova2016knowledge}.

%%%%%%%%%%%%%%%%%%%%%%%%%%%%%%%%%%%%%%%%%%%%%%%%%%%%%%%%%%%%%%%%%%%%%%%%%%%%%%%%
\subsection{Algebraic Decision Diagrams}

The central data structure we use in this work is \textdef{Algebraic Decision Diagram (ADD)} \cite{BFGHMPS97}, a compact representation of a function as a directed acyclic graph.
Formally, an ADD is a tuple $(X, S, \pi, G)$, where $X$ is a set of Boolean variables, $S$ is an arbitrary set (called the \textdef{carrier set}), $\pi: X \rightarrow \N$ is an injection (called the \textdef{diagram variable order}), and $G$ is a rooted directed acyclic graph satisfying the following three properties.
First, every terminal node of $G$ is labeled with an element of $S$.
Second, every non-terminal node of $G$ is labeled with an element of $X$ and has two outgoing edges labeled 0 and 1.
Finally, for every path in $G$, the labels of the visited non-terminal nodes must occur in increasing order under $\pi$.
ADDs were originally designed for matrix multiplication and shortest path algorithms \cite{BFGHMPS97}.
ADDs have also been used for stochastic model checking \cite{KNP07} and stochastic planning \cite{HSHB99}.
In this work, we do not need arbitrary carrier sets; it is sufficient to consider ADDs with $S = \mathbb{R}$.

An ADD $(X, S, \pi, G)$ is a compact representation of a function $f: 2^X \rightarrow S$.
Although there are many ADDs representing each such function $f$, for each injection $\pi: X \rightarrow \N$, there is a unique minimal ADD that represents $f$ with $\pi$ as the diagram variable order, called the \textdef{canonical ADD}.
ADDs can be minimized in polynomial time, so it is typical to only work with canonical ADDs.
Given two ADDs representing functions $f$ and $g$, the ADDs representing $f + g$ and $f \mult g$ can also be computed in polynomial time.

The choice of diagram variable order can have a dramatic impact on the size of the ADD.
A variety of techniques exist to heuristically find diagram variable orders.
Moreover, since Binary Decision Diagrams (BDDs) \cite{bryant86} can be seen as ADDs with carrier set $S = \set{0, 1}$, there is significant overlap with the techniques to find variable orders for BDDs.
% We discuss these heuristics in more detail in Section \ref{sec_algorithm}.

Several packages exist for efficiently manipulating ADDs.
Here we use the package \cudd{} \cite{somenzi09}, which supports carrier sets $S = \set{0, 1}$ and (using floating-point arithmetic) $S = \mathbb{R}$.
\cudd{} implements several ADD operations, such as addition, multiplication, and projection.

%%%%%%%%%%%%%%%%%%%%%%%%%%%%%%%%%%%%%%%%%%%%%%%%%%%%%%%%%%%%%%%%%%%%%%%%%%%%%%%%
%% file end 2-preliminaries.tex
%%%%%%%%%%%%%%%%%%%%%%%%%%%%%%%%%%%%%%%%%%%%%%%%%%%%%%%%%%%%%%%%%%%%%%%%%%%%%%%%

%%%%%%%%%%%%%%%%%%%%%%%%%%%%%%%%%%%%%%%%%%%%%%%%%%%%%%%%%%%%%%%%%%%%%%%%%%%%%%%%
%% file start 3-theory.tex
%%%%%%%%%%%%%%%%%%%%%%%%%%%%%%%%%%%%%%%%%%%%%%%%%%%%%%%%%%%%%%%%%%%%%%%%%%%%%%%%

\section{Using ADDs for Weighted Model Counting with Early Projection}
\label{sec_theory}

An ADD with carrier set $\R$ can be used to represent both a Boolean formula $\varphi: 2^X \to \set{0, 1}$ and a weight function $W: 2^X \to \R$.
ADDs are thus a natural candidate as a data structure for weighted model counting algorithms.

In this section, we outline theoretical foundations for performing weighted model counting with ADDs.
We consider first the general case of weighted model counting.
We then specialize to literal-weighted model counting of CNF formulas and show how the technique of early projection can take advantage of such factored representations of Boolean formulas $\varphi$ and weight functions $W$.

%%%%%%%%%%%%%%%%%%%%%%%%%%%%%%%%%%%%%%%%%%%%%%%%%%%%%%%%%%%%%%%%%%%%%%%%%%%%%%%%
\subsection{General Weighted Model Counting}

We assume that the Boolean formula $\varphi$ and the weight function $W$ are represented as ADDs.
The goal is to compute $W(\varphi)$, the weighted model count of $\varphi$ \wrt{} $W$.
To do this, we define two operations on functions $2^X \to \R$ that can be efficiently computed using the ADD representation: \textdef{product} and \textdef{projection}.
These operations are combined in Theorem \ref{theorem_wmc} to perform weighted model counting.

First, we define the \textdef{product} of two functions.
\begin{definition}
\label{def_add_mult}
  Let $X$ and $Y$ be sets of variables.
  The \textdef{product} of functions $A: 2^X \to \R$ and $B: 2^Y \to \R$ is the function $A \mult B: 2^{X \cup Y} \to \R$ defined for all $\tau \in 2^{X \cup Y}$ by
  $$(A \mult B)(\tau) = A(\tau \cap X) \mult B(\tau \cap Y).$$
\end{definition}

Note that the operator $\mult$ is commutative and associative, and it has the identity element $\mathbf{1}: 2^\sete \to \R$ (that maps $\sete$ to $1$).
If $\varphi: 2^X \to \set{0, 1}$ and $\psi: 2^Y \to \set{0, 1}$ are Boolean formulas, the product $\varphi \mult \psi$ is the Boolean function corresponding to the conjunction $\varphi \wedge \psi$.
% As mentioned earlier, \tool{CUDD} has an implementation of the operator $\mult$ for ADDs.

Second, we define the \textdef{projection} of a Boolean variable $x$ in a real-valued function $A$, which reduces the number of variables in $A$ by ``summing out'' $x$.
Variable projection in real-valued functions is similar to variable elimination in Bayesian networks \cite{zhang1994simple}.
% The latter ``sums out'' a variable $v$ from a set $S$ of potentials by: removing all potentials $p_i$ containing $v$ from $S$, calculating $p' = \sum_{v \in \set{0, 1}} \pars{\prod_i p_i}$, and adding the potential $p'$ to $S$.

\begin{definition}
\label{def_add_project}
  Let $X$ be a set of variables and $x \in X$.
  The \textdef{projection} of $A: 2^X \to \R$ \wrt{} $x$ is the function $\exists_x A: 2^{X \setminus \set{x}} \to \R$ defined for all $\tau \in 2^{X \setminus \set{x}}$ by
  $$(\exists_x A)(\tau) = A(\tau) + A(\tau \cup \set{x}).$$
\end{definition}

One can check that projection is commutative, \ie, that $\exists_x \exists_y A = \exists_y \exists_x A$ for all variables $x, y \in X$ and functions $A: 2^X \to \R$.
If $X = \set{x_1, x_2, \ldots, x_n}$, define
$$\exists_X A = \exists_{x_1} \exists_{x_2} \ldots \exists_{x_n} A.$$
% \tool{CUDD} has an implementation of the projection operator for ADDs.

We are now ready to use product and projection to do weighted model counting, through the following theorem.
\begin{theorem}
\label{theorem_wmc}
  Let $\varphi: 2^X \to \set{0, 1}$ be a Boolean formula over a set $X$ of variables, and let $W: 2^X \to \R$ be an arbitrary weight function.
  Then
  $$W(\varphi) = (\exists_X (\varphi \mult W))(\sete).$$
\end{theorem}

Theorem \ref{theorem_wmc} suggests that $W(\varphi)$ can be computed by constructing an ADD for $\varphi$ and another for $W$, computing the ADD for their product $\varphi \mult W$, and performing a sequence of projections to obtain the final weighted model count.
Unfortunately, this ``monolithic'' approach is infeasible in most interesting cases: the ADD representation of $\varphi \mult W$ is often too large, even with the best possible diagram variable order.

Instead, we next show a technique for avoiding the construction of an ADD for $\varphi \mult W$ by rearranging the products and projections.

%%%%%%%%%%%%%%%%%%%%%%%%%%%%%%%%%%%%%%%%%%%%%%%%%%%%%%%%%%%%%%%%%%%%%%%%%%%%%%%%
\subsection{Early Projection}

A key technique in symbolic computation is \textdef{early projection}: when performing a product followed by a projection (as in Theorem \ref{theorem_wmc}), it is sometimes possible and advantageous to perform the projection first.
Early projection is possible when one component of the product does not depend on the projected variable.
Early projection has been used in a variety of settings, including database-query optimization \cite{KV00b}, symbolic model checking \cite{BCL91}, and satisfiability solving \cite{PV05}.
The formal statement is as follows.

\begin{theorem}[Early Projection]
\label{theorem_early_project}
  Let $X$ and $Y$ be sets of variables, $A: 2^X \to \R$, and $B: 2^Y \to \R$.
  For all $x \in X \setminus Y$,
  $$\exists_x (A \mult B) = \pars{\exists_x A} \mult B.$$
  As a corollary, for all $X' \subseteq X \setminus Y$,
  $$\exists_{X'} (A \mult B) = \pars{\exists_{X'} A} \mult B.$$
\end{theorem}

The use of early projection in Theorem \ref{theorem_wmc} is quite limited when $\varphi$ and $W$ have already been represented as ADDs, since on many benchmarks both $\varphi$ and $W$ depend on most of the variables.
If $\varphi$ is a CNF formula and $W$ is a literal-weight function, however, both $\varphi$ and $W$ can be rewritten as products of smaller functions.
This can significantly increase the applicability of early projection.

Assume that $\varphi : 2^X \to \set{0, 1}$ is a CNF formula, \ie, given as a set of clauses.
For every clause $\gamma \in \varphi$, observe that $\gamma$ is a Boolean formula $\gamma: 2^{X_\gamma} \to \set{0, 1}$ where $X_\gamma \subseteq X$ is the set of variables appearing in $\gamma$.
One can check using Definition \ref{def_add_mult} that $\varphi = \prod_{\gamma \in \varphi} \gamma$.

Similarly, assume that $W: 2^X \to \R$ is a literal-weight function.
For every $x \in X$, define $W_x : 2^{\set{x}} \to \R$ to be the function that maps $\sete$ to $W^-(x)$ and $\set{x}$ to $W^+(x)$.
One can check using Definition \ref{def_add_mult} that $W = \prod_{x \in X} W_x$.

When $\varphi$ is a CNF formula and $W$ is a literal-weight function, we can rewrite Theorem \ref{theorem_wmc} as
\begin{equation}
\label{eqn_computing_lwmc}
  W(\varphi) = \pars{\exists_X \pars{\prod_{\gamma \in \varphi} \gamma \mult \prod_{x \in X} W_x}}(\sete).
\end{equation}

By taking advantage of the associative and commutative properties of multiplication as well as the commutative property of projection, it is possible to rearrange Equation \ref{eqn_computing_lwmc} in order to apply early projection.
We present an algorithm to perform this rearrangement in the following section.

%%%%%%%%%%%%%%%%%%%%%%%%%%%%%%%%%%%%%%%%%%%%%%%%%%%%%%%%%%%%%%%%%%%%%%%%%%%%%%%%
%% file end 3-theory.tex
%%%%%%%%%%%%%%%%%%%%%%%%%%%%%%%%%%%%%%%%%%%%%%%%%%%%%%%%%%%%%%%%%%%%%%%%%%%%%%%%

%%%%%%%%%%%%%%%%%%%%%%%%%%%%%%%%%%%%%%%%%%%%%%%%%%%%%%%%%%%%%%%%%%%%%%%%%%%%%%%%
%% file start 4-algorithm.tex
%%%%%%%%%%%%%%%%%%%%%%%%%%%%%%%%%%%%%%%%%%%%%%%%%%%%%%%%%%%%%%%%%%%%%%%%%%%%%%%%

%%%%%%%%%%%%%%%%%%%%%%%%%%%%%%%%%%%%%%%%%%%%%%%%%%%%%%%%%%%%%%%%%%%%%%%%%%%%%%%%

\newcommand{\algoAddmc}{
\begin{algorithm*}
\label{algo_lwmc_cnf}
  \caption{CNF literal-weighted model counting with ADDs}
  \DontPrintSemicolon
  \KwIn{$X$: set of Boolean variables}
  \KwIn{$\varphi$: nonempty CNF formula over $X$}
  \KwIn{$W$: literal-weight function over $X$}
  \KwOut{$W(\varphi)$: weighted model count of $\varphi$ \wrt{} $W$}
  % \Begin{
    $\pi \gets \getDiagramVarOrder(\varphi)$
      \tcc*{injection $\pi : X \to \N$}
    $\rho \gets \getClusterVarOrder(\varphi)$
      \tcc*{injection $\rho : X \to \N$}
    $m \gets \max_{x \in X} \rho(x)$ \\
    \For{$i = m, m - 1, \ldots, 1$}{
      $\Gamma_i \gets \set{\gamma \in \varphi : \get{clause-rank}(\gamma, \rho) = i}$
        \tcc*{collecting clauses $\gamma$ with rank $i$}
      $\kappa_i \gets \set{\get{clause-ADD}(\gamma, \pi) : \gamma \in \Gamma_i}$
        \tcc*{cluster $\kappa_i$ contains ADDs of clauses $\gamma$ with rank $i$}
      $X_i \gets \vars(\kappa_i) \setminus \bigcup^m_{p = i + 1} \vars(\kappa_p)$
        \tcc*{variables already placed in $X_i$ will not be placed in $X_1, X_2, \ldots, X_{i - 1}$}
    }
    % $X_m \gets X_m \cup (X \setminus \vars(\varphi))$
    %   \tcc*{collecting variables declared in $X$ but do not appear in $\varphi$}
    \For{$i = 1, 2, \ldots, m$}{ \label{line_loop2}
      \If{$\kappa_i \ne \sete$}{
        $A_i \gets \prod_{D \in \kappa_i} D$
          \tcc*{product of all ADDs $D$ in cluster $\kappa_i$}
        \For{$x \in X_i$}{ \label{line_loop2_inner}
          $A_i \gets \exists_x \pars{A_i \mult W_x}$
            \tcc*{$W_x : 2^{\set{x}} \to \R$, represented by an ADD}
        }
        \If{$i < m$}{
          $j \gets \chooseCluster(A_i, i)$ \label{line_choose_cluster}
            \tcc*{$i < j \le m$}
          $\kappa_j \gets \kappa_j \cup \set{A_i}$ \label{line_add_to_cluster}
        }
      }
    }
    \Return{$A_m(\sete)$} \label{line_return}
      \tcc*{$A_m : 2^\sete \to \R$}
  % }
\end{algorithm*}
}

%%%%%%%%%%%%%%%%%%%%%%%%%%%%%%%%%%%%%%%%%%%%%%%%%%%%%%%%%%%%%%%%%%%%%%%%%%%%%%%%

\algoAddmc

\section{Dynamic Programming for Model Counting}
\label{sec_algorithm}

In this section, we discuss an algorithm for performing literal-weighted model counting of CNF formulas using ADDs through dynamic-programming techniques.

Our algorithm is presented as Algorithm \ref{algo_lwmc_cnf}.
Broadly, our algorithm partitions the clauses of a formula $\varphi$ into clusters.
For each cluster, we construct the ADD corresponding to the conjunction of its clauses.
These ADDs are then incrementally combined via the multiplication operator.
Throughout, each variable of the ADD is projected as early as Theorem \ref{theorem_early_project} allows ($X_i$ is the set of variables that can be projected in each iteration $i$ of the second loop).
At the end of the algorithm, all variables have been projected, and the resulting ADD has a single node representing the weighted model count.
This algorithm can be seen as rearranging the terms of Equation \ref{eqn_computing_lwmc} (according to the clusters) in order to perform the projections indicated by $X_i$ at each step $i$.

The function $\get{clause-ADD}(\gamma, \pi)$ constructs the ADD representing the clause $\gamma$ using the diagram variable order $\pi$.
The remaining functions that appear throughout Algorithm \ref{algo_lwmc_cnf}, namely $\getDiagramVarOrder$, $\getClusterVarOrder$, $\get{clause-rank}$, and $\chooseCluster$, represent heuristics that can be used to tune the specifics of the algorithm.

Before discussing the various heuristics, we assert the correctness of Algorithm \ref{algo_lwmc_cnf} in the following theorem.
\begin{theorem}
\label{theorem_algo_correct}
  Let $X$ be a set of variables, $\varphi$ be a nonempty CNF formula over $X$, and $W$ be a literal-weight function over $X$.
  Assume that $\getDiagramVarOrder$ and $\getClusterVarOrder$ return injections $X \rightarrow \N$.
  Furthermore, assume that all $\get{clause-rank}$ and $\chooseCluster$ calls satisfy the following conditions:
  \begin{enumerate}[ref=\arabic*]
    \item $1 \le \get{clause-rank}(\gamma, \rho) \le m$, \label{cond1}
    \item $i < \chooseCluster(A_i, i) \le m$, and \label{cond2}
    \item $X_s \cap \vars(A_i) = \sete$ for all integers $s$ such that $i < s < \chooseCluster(A_i, i)$. \label{cond3}
  \end{enumerate}
  Then Algorithm \ref{algo_lwmc_cnf} returns $W(\varphi)$.
\end{theorem}

By Condition \ref{cond1}, we know the set $\set{\Gamma_1, \Gamma_2, \ldots, \Gamma_m}$ forms a partition of the clauses in $\varphi$.
Condition \ref{cond2} ensures that lines \ref{line_choose_cluster}-\ref{line_add_to_cluster} place $A_i$ in a cluster that has not yet been processed.
Also on lines \ref{line_choose_cluster}-\ref{line_add_to_cluster}, Condition \ref{cond3} prevents $A_i$ from skipping a cluster $\kappa_s$ which shares some variable $y$ with $A_i$, as $y$ will be projected at step $s$.
These three invariants are sufficient to prove that Algorithm \ref{algo_lwmc_cnf} indeed computes the weighted model count of $\varphi$ \wrt{} $W$.
All heuristics we describe in this paper satisfy the conditions of Theorem \ref{theorem_algo_correct}.

In the remainder of this section, we discuss a variety of existing heuristics that can be used as a part of Algorithm \ref{algo_lwmc_cnf} to implement $\getDiagramVarOrder$, $\getClusterVarOrder$, $\get{clause-rank}$, and $\chooseCluster$.
% We compare the performance of these heuristics for Algorithm \ref{algo_lwmc_cnf} in Section \ref{sec_experiments}.

%%%%%%%%%%%%%%%%%%%%%%%%%%%%%%%%%%%%%%%%%%%%%%%%%%%%%%%%%%%%%%%%%%%%%%%%%%%%%%%%
\subsection{Heuristics for $\getDiagramVarOrder$ and $\getClusterVarOrder$}

The heuristic chosen for $\getDiagramVarOrder$ indicates the variable order that is used as the diagram variable order in every ADD constructed by Algorithm \ref{algo_lwmc_cnf}.
The heuristic chosen for $\getClusterVarOrder$ indicates the variable order which, when combined with the heuristic for $\get{clause-rank}$, is used to order the clauses of $\varphi$.
(\heuristic{BE} orders clauses by the smallest variable that appears in each clause, while \heuristic{BM} orders clauses by the largest variable.)
In this work, we consider seven possible heuristics for each variable order: \heuristic{Random}, \heuristic{MCS}, \heuristic{LexP}, \heuristic{LexM}, \heuristic{InvMCS}, \heuristic{InvLexP}, and \heuristic{InvLexM}.

One simple heuristic for $\getDiagramVarOrder$ and $\getClusterVarOrder$ is to randomly order the variables, \ie, for a formula over some set $X$ of variables, sample an injection $X \rightarrow \set{1, 2, \ldots, |X|}$ uniformly at random.
We call this the \textbf{Random} heuristic.
\textbf{Random} is a baseline for comparison of the other variable order heuristics.

For the remaining heuristics, we must define the \textdef{Gaifman graph} $G_\varphi$ of a formula $\varphi$.
The Gaifman graph of $\varphi$ has a vertex for every variable in $\varphi$.
Two vertices are connected by an edge if and only if the corresponding variables appear in the same clause of $\varphi$.

One such heuristic is called \textdef{Maximum-Cardinality Search} \cite{tarjan1984simple}.
At each step of the heuristic, the next variable chosen is the variable adjacent in $G_\varphi$ to the greatest number of previously chosen variables (breaking ties arbitrarily).
We call this the \heuristic{MCS} heuristic for variable order.

Another such heuristic is called \textdef{Lexicographic search for perfect orders} \cite{koster2001treewidth}.
Each vertex of $G_\varphi$ is assigned an initially-empty set of vertices (called the \textdef{label}).
At each step of the heuristic, the next variable chosen is the variable $x$ whose label is lexicographically smallest among the unchosen variables (breaking ties arbitrarily).
Then $x$ is added to the label of its neighbors in $G_\varphi$.
We call this the \heuristic{LexP} heuristic for variable order.

A similar heuristic is called \textdef{Lexicographic search for minimal orders} \cite{koster2001treewidth}.
As before, each vertex of $G_\varphi$ is assigned an initially-empty label.
At each step of the heuristic, the next variable chosen is again the variable $x$ whose label is lexicographically smallest (breaking ties arbitrarily).
In this case, $x$ is added to the label of every variable $y$ where there is a path $x, z_1, z_2, \ldots, z_k, y$ in $G_\varphi$ such that every $z_i$ is unchosen and the label of $z_i$ is lexicographically smaller than the label of $y$.
We call this the \heuristic{LexM} heuristic for variable order.

Additionally, the variable orders produced by each of the heuristics \heuristic{MCS}, \heuristic{LexP}, and \heuristic{LexM} can be inverted.
We call these new heuristics \heuristic{InvMCS}, \heuristic{InvLexP}, and \heuristic{InvLexM}.

%%%%%%%%%%%%%%%%%%%%%%%%%%%%%%%%%%%%%%%%%%%%%%%%%%%%%%%%%%%%%%%%%%%%%%%%%%%%%%%%
\subsection{Heuristics for $\get{clause-rank}$}

The heuristic chosen for $\get{clause-rank}$ indicates the strategy used for clustering the clauses of $\varphi$.
In this work, we consider three possible heuristics to be chosen for $\get{clause-rank}$ that satisfy the conditions of Theorem \ref{theorem_algo_correct}: \heuristic{Mono}, \heuristic{BE}, and \heuristic{BM}.

One simple case is when the rank of each clause is constant, \eg, when $\get{clause-rank}(\gamma, \rho) = m$ for all $\gamma \in \varphi$.
In this case, all clauses of $\varphi$ are placed in $\Gamma_m$, so Algorithm \ref{algo_lwmc_cnf} combines all clauses of $\varphi$ into a single ADD before performing projections.
This corresponds to the monolithic approach we mentioned earlier.
We thus call this the \heuristic{Mono} heuristic for $\get{clause-rank}$.
Notice that the performance of Algorithm \ref{algo_lwmc_cnf} with \heuristic{Mono} does not depend on the heuristic for $\getClusterVarOrder$ or $\chooseCluster$.
This heuristic has previously been applied to ADDs in the context of knowledge compilation \cite{fargier2014knowledge}.

% A different heuristic puts each clause in its own cluster, \eg, $\get{clause-rank}(\gamma_i, \rho) = i$ for each $\gamma_i \in \varphi$. This is the \heuristic{Linear} heuristic.

Another, more complex heuristic assigns the rank of each clause to be the smallest $\rho$-rank of the variables that appear in the clause.
That is, $\get{clause-rank}(\gamma, \rho) = \min_{x \in \vars(\gamma)} \rho(x)$.
This heuristic corresponds to \textdef{bucket elimination} \cite{dechter99}, so we call this the \heuristic{BE} heuristic.
In this case, notice that every clause containing $x \in X$ can only appear in $\Gamma_i$ with $i \le \rho(x)$.
It follows that $x$ has always been projected from all clauses by the end of iteration $\rho(x)$ in the second loop of Algorithm \ref{algo_lwmc_cnf} using \heuristic{BE}.

A related heuristic assigns the rank of each clause to be the largest $\rho$-rank of the variables that appear in the clause.
That is, $\get{clause-rank}(\gamma, \rho) = \max_{x \in \vars(\gamma)} \rho(x)$.
This heuristic corresponds to \textdef{Bouquet's Method} \cite{bouquet99}, so we call this the \heuristic{BM} heuristic.
Unlike the \heuristic{BE} case, we can make no guarantee about when each variable is projected in Algorithm \ref{algo_lwmc_cnf} using \heuristic{BM}.

%%%%%%%%%%%%%%%%%%%%%%%%%%%%%%%%%%%%%%%%%%%%%%%%%%%%%%%%%%%%%%%%%%%%%%%%%%%%%%%%
\subsection{Heuristics for $\chooseCluster$}

The heuristic chosen for $\chooseCluster$ indicates the strategy for combining the ADDs produced from each cluster.
In this work, we consider two possible heuristics to use for $\chooseCluster$ that satisfy the conditions of Theorem \ref{theorem_algo_correct}: \heuristic{List} and \heuristic{Tree}.

One heuristic is when $\chooseCluster$ selects to place $A_i$ in the closest cluster that satisfies the conditions of Theorem \ref{theorem_algo_correct}, namely the next cluster to be processed.
That is, $\chooseCluster(A_i, i) = i + 1$.
Under this heuristic, Algorithm \ref{algo_lwmc_cnf} equivalently builds an ADD for each cluster and then combines the ADDs in a one-by-one, in-order fashion, projecting variables as early as possible.
In every iteration, there is a single intermediate ADD representing the combination of previous clusters.
We call this the \heuristic{List} heuristic.

Another heuristic is when $\chooseCluster$ selects to place $A_i$ in the furthest cluster that satisfies the conditions of Theorem \ref{theorem_algo_correct}.
That is, $\chooseCluster(A_i, i)$ returns the smallest $j > i$ such that $X_j \cap \vars(A_i) \ne \sete$ (or returns $m$, if $\vars(A_i) = \sete)$.
In every iteration, there may be multiple intermediate ADDs representing the combinations of previous clusters.
We call this the \heuristic{Tree} heuristic.

Notice that the computational structure of Algorithm \ref{algo_lwmc_cnf} can be represented by a tree of clusters, where cluster $\kappa_i$ is a child of cluster $\kappa_j$ whenever the ADD produced from $\kappa_i$ is placed in $\kappa_j$ (lines \ref{line_choose_cluster}-\ref{line_add_to_cluster}).
These trees are always left-deep under the \heuristic{List} heuristic, but they can be more complex under the \heuristic{Tree} heuristic.

We can combine $\get{clause-rank}$ heuristics and (if applicable) $\chooseCluster$ heuristics to form \textdef{clustering heuristics}: $\Mono$, $\Be-\List$, $\Be-\Tree$, $\Bm-\List$, and $\Bm-\Tree$.

%%%%%%%%%%%%%%%%%%%%%%%%%%%%%%%%%%%%%%%%%%%%%%%%%%%%%%%%%%%%%%%%%%%%%%%%%%%%%%%%
%% file end 4-algorithm.tex
%%%%%%%%%%%%%%%%%%%%%%%%%%%%%%%%%%%%%%%%%%%%%%%%%%%%%%%%%%%%%%%%%%%%%%%%%%%%%%%%

%%%%%%%%%%%%%%%%%%%%%%%%%%%%%%%%%%%%%%%%%%%%%%%%%%%%%%%%%%%%%%%%%%%%%%%%%%%%%%%%
%% file start 5-experiments.tex
%%%%%%%%%%%%%%%%%%%%%%%%%%%%%%%%%%%%%%%%%%%%%%%%%%%%%%%%%%%%%%%%%%%%%%%%%%%%%%%%

\section{Empirical Evaluation}
\label{sec_experiments}

%%%%%%%%%%%%%%%%%%%%%%%%%%%%%%%%%%%%%%%%%%%%%%%%%%%%%%%%%%%%%%%%%%%%%%%%%%%%%%%%

\newcommand{\figwidth}{.95\columnwidth} % AAAI style

%%%%%%%%%%%%%%%%%%%%%%%%%%%%%%%%%%%%%%%%%%%%%%%%%%%%%%%%%%%%%%%%%%%%%%%%%%%%%%%%

We implement Algorithm \ref{algo_lwmc_cnf} using the ADD package \cudd{} to produce \addmc{}, a new weighted model counter.
\addmc{} supports all heuristics described in Section \ref{sec_algorithm}.
The \addmc{} source code and experimental data can be obtained from a public repository (\url{https://github.com/vardigroup/ADDMC}).

We aim to: (1) find good heuristic configurations for our tool \addmc, and (2) compare \addmc{} against four state-of-the-art weighted model counters: \cachet{} \cite{SBBK04}, \ctd{} \cite{darwiche2004new}, \df{} \cite{lagniez2017improved}, and \minictd{} \cite{OD15}.
To accomplish this, we use a set of \benchmarkCountAltogether{} CNF literal-weighted model counting benchmarks, which were gathered from two sources.

First, the \classBayes{} class%
\footnote{\urlBenchmarksBayes}
contains \benchmarkCountBayes{} benchmarks.
The application domain is Bayesian inference \cite{sang2005solving}.
% This benchmark class is subdivided into three families: \famDqmr, \famGrid, and \famPlanRec.
The accompanied literal weights are in the interval $[0, 1]$.

Second, the \classPseudoweighted{} class%
\footnote{\urlBenchmarksPseudoweighted}
contains \benchmarkCountPseudoweighted{} benchmarks.
This benchmark class is subdivided into eight families: \famBmc, \famCircuit, \famConfig, \famHandmade, \famPlanning, \famQif, \famRandom, and \famSchedule{} \cite{clarke2001bounded,sinz2003formal,palacios2009compiling,klebanov2013sat}.
All of these benchmarks are originally unweighted.
As we focus in this work on weighted model counting, we generate weights by, for each variable $x$, randomly assigning: either weights $W^+(x) = 0.5$ and $W^-(x) = 1.5$, or $W^+(x) = 1.5$ and $W^-(x) = 0.5$.%
\footnote{
  For each variable $x$, \cachet{} requires $W^+(x) + W^-(x) = 1$ unless $W^+(x) = W^-(x) = 1$.
  So we use weights 0.25 and 0.75 for \cachet{} and multiply the model count produced by \cachet{} on a formula $\varphi$ by $2^{\size{\vars(\varphi)}}$ as a postprocessing step.
}
Generating weights in this particular fashion results in a reasonably low amount of floating-point underflow and overflow for all model counters.

%%%%%%%%%%%%%%%%%%%%%%%%%%%%%%%%%%%%%%%%%%%%%%%%%%%%%%%%%%%%%%%%%%%%%%%%%%%%%%%%
\subsection{Experiment 1: Comparing \addmc{} Heuristics}
\label{sec_experiments_heuristics}

\addmc{} heuristic configurations are constructed from five clustering heuristics (\Mono, \Be-\List, \Be-\Tree, \Bm-\List, and \Bm-\Tree) together with seven variable order heuristics (\Random, \Mcs, \Invmcs, \Lexp, \Invlexp, \Lexm, and \Invlexm).
Using one variable order heuristic for the cluster variable order and another for the diagram variable order gives us \heuristicConfigCount{} configurations in total.
We compare these configurations to find the best combination of heuristics.

On a Linux cluster with Xeon E5-2650v2 CPUs (2.60-GHz), we run each combination of heuristics on each benchmark using a single core, 24 GB of memory, and a \timeoutAllHeuristics-second timeout.

%%%%%%%%%%%%%%%%%%%%%%%%%%%%%%%%%%%%%%%%%%%%%%%%%%%%%%%%%%%%%%%%%%%%%%%%%%%%%%%%
\subsubsection{Performance Analysis}

\begin{table}
  \caption{The numbers of benchmarks solved (of \benchmarkCountAltogether) in \timeoutAllHeuristics{} seconds by the best, second-best, median, and worst \addmc{} heuristic configurations.}
  \centering
  \begin{tabular}{|l|l|l|r|l|} \hline
    Clustering & Clus. var. & Diag. var. & Solved & Name \\ \hline
    \Bm-\Tree & \Lexp & \Mcs & 1202 & \besto{} \\ \hline
    \Be-\Tree & \Invlexp & \Mcs & 1200 & \bestt{} \\ \hline
    \Be-\List & \Lexp & \Lexp & 504 & \textit{Median} \\ \hline
    % \Mono & -- & \Lexp & 188 & \textit{Best-Mono} \\ \hline
    % \multicolumn{2}{|l|}{\Mono} & \Lexp & 188 & \textit{Best-Mono} \\ \hline
    \Be-\List & \Random & \Random & 53 & \textit{Worst} \\ \hline
  \end{tabular}
\label{table_heuristics}
\end{table}

Table \ref{table_heuristics} reports the numbers of benchmarks solved by four \addmc{} heuristic configurations: best, second-best, median, and worst (of \heuristicConfigCount{} configurations in total).
Bouquet's Method (\Bm) and bucket elimination (\Be) have similar-performing top configurations: \besto{} and \bestt{}.
This shows that Bouquet's Method is competitive with bucket elimination.

\begin{figure}
  \includegraphics[width=\figwidth]{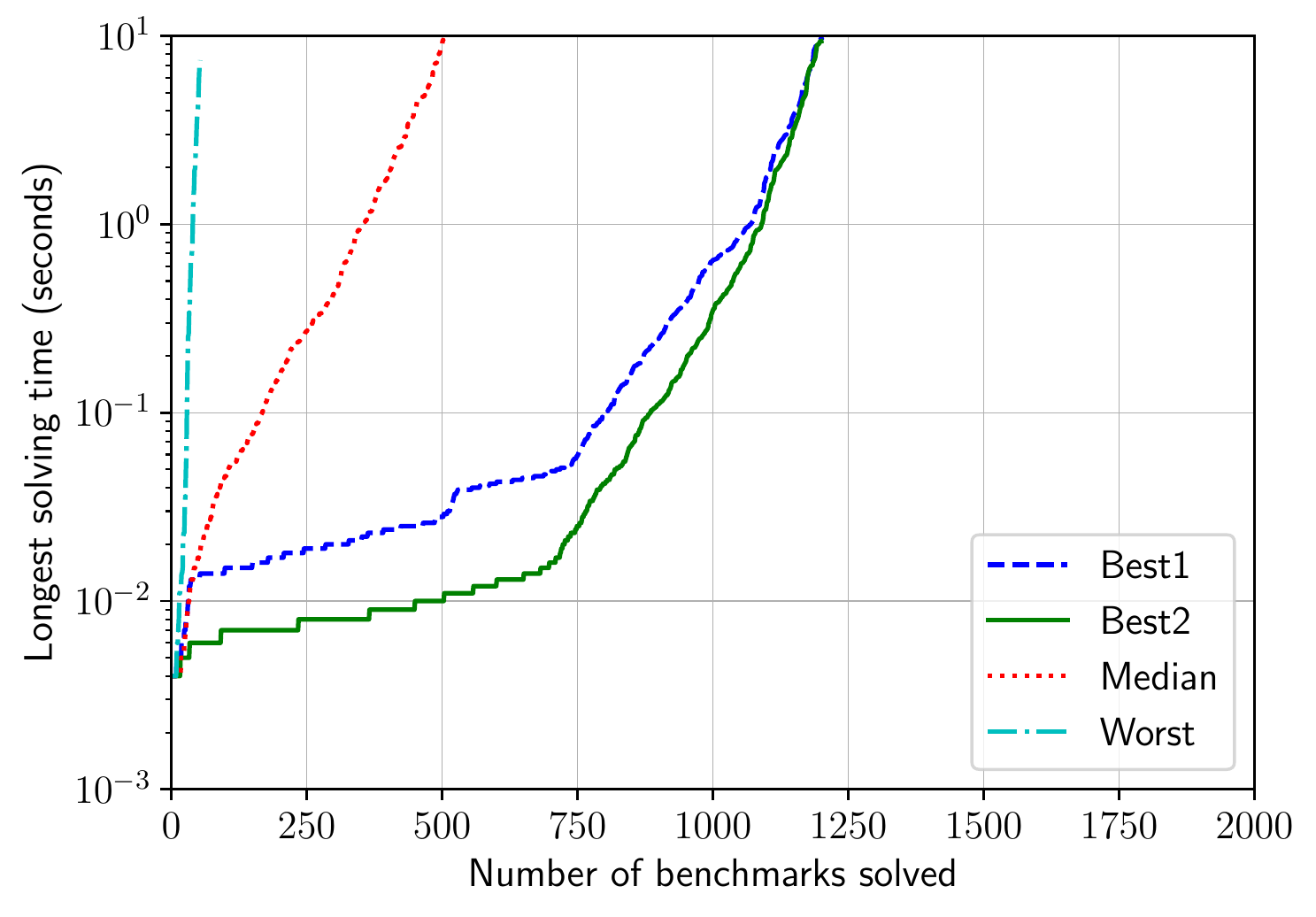}
  \caption{A cactus plot of the numbers of benchmarks solved by the best, second-best, median, and worst \addmc{} heuristic configurations.}
\label{fig_heuristics}
\end{figure}

See Figure \ref{fig_heuristics} for a more detailed analysis of the runtime of these four heuristic configurations.
Evidently, some configurations perform quite well while others perform quite poorly.
The wide range of performance indicates that the choice of heuristics is essential to the competitiveness of \addmc{}.

We choose \besto{} (\Bm-\Tree{} clustering with \Lexp{} cluster variable order and \Mcs{} diagram variable order), which was the heuristic configuration able to solve the most benchmarks within \timeoutAllHeuristics{} seconds, as the representative \addmc{} configuration for Experiment 2.

%%%%%%%%%%%%%%%%%%%%%%%%%%%%%%%%%%%%%%%%%%%%%%%%%%%%%%%%%%%%%%%%%%%%%%%%%%%%%%%%
\subsection{Experiment 2: Comparing Weighted Model Counters}
\label{sec_experiments_counters}

In the previous experiment, the \addmc{} heuristic configuration able to solve the most benchmarks is \besto{} (\Bm-\Tree{} clustering with \Lexp{} cluster variable order and \Mcs{} diagram variable order).
Using this configuration, we now compare \addmc{} to four state-of-the-art weighted model counters: \cachet, \ctd%
\footnote{
  \ctd{} does not natively support weighted model counting.
  To compare \ctd{} to weighted model counters, we use \ctd{} to compile CNF into d-DNNF then use \reasoner{} (\url{http://www.cril.univ-artois.fr/kc/d-DNNF-reasoner.html}) to compute the weighted model count.
  On average, \ctd's compilation time is 81.65\% of the total time.
}%
, \df, and \minictd.
(We note that \cachet{} uses \texttt{long double} precision, whereas all other model counters use \texttt{double} precision.)

On a Linux cluster with Xeon E5-2650v2 CPUs (2.60-GHz), we run each counter on each benchmark using a single core, 24 GB of memory and a \timeoutCounters-second timeout.

%%%%%%%%%%%%%%%%%%%%%%%%%%%%%%%%%%%%%%%%%%%%%%%%%%%%%%%%%%%%%%%%%%%%%%%%%%%%%%%%
\subsubsection{Correctness Analysis}

To compare answers computed by different weighted model counters (in the presence of imprecision from floating-point arithmetic), we consider non-negative real numbers $a \le b$ equal when:
$b - a \le 10^{-3}$ if $a = 0$ or $b \le 1$, and $b / a \le 1 + 10^{-3}$ otherwise.
% \begin{align}
% \label{tolerance}
%   \begin{cases}
%     b - a \le 10^{-3} & \text{if } a = 0 \text{ or } b \le 1 \\
%     b / a \le 1 + 10^{-3} & \text{otherwise}
%   \end{cases}
% \end{align}

Even with this equality tolerance, weighted model counters still sometimes produce different answers for the same benchmark due to floating-point effects.
In particular, of 1008 benchmarks that are solved by all five model counters, \addmc{} produces 7 model counts that differ from the output of all four other tools.
For \cachet, \ctd, \df, and \minictd, the numbers are respectively 55, 0, 42, and 0. % normalization of "zero" to "0"
To improve \addmc's precision, we plan as future work to integrate a new decision diagram package, \sylvan{} \cite{van2015sylvan}, into \addmc{}.
\sylvan{} can interface with the GNU Multiple Precision library to support ADDs with higher-precision numbers.

%%%%%%%%%%%%%%%%%%%%%%%%%%%%%%%%%%%%%%%%%%%%%%%%%%%%%%%%%%%%%%%%%%%%%%%%%%%%%%%%
\subsubsection{Performance Analysis}

\begin{table}
  \caption{The numbers of benchmarks solved (of \benchmarkCountAltogether) in \timeoutCounters{} seconds --- uniquely (\ie, benchmarks solved by no other solver), fastest, and in total --- by five weighted model counters and two virtual best solvers (\vbs{} and \vbss{}).}
  \centering
  \begin{tabular}{|l|r|r|r|} \hline
    \multirow{2}{*}{Solvers} & \multicolumn{3}{c|}{Benchmarks solved} \\ \cline{2-4}
    & Unique & Fastest & Total \\ \hline
    \vbs{} (with \addmc) & -- & -- & 1771 \\
    \vbss{} (without \addmc) & -- & -- & 1647 \\ \hline
    \df & 12 & 283 & 1587 \\
    \ctd & 0 & 13 & 1417 \\
    \minictd & 8 & 61 & 1407 \\
    \addmc{} & 124 & 763 & 1404 \\
    \cachet & 14 & 651 & 1383 \\ \hline
  \end{tabular}
\label{table_counters_altogether}
\end{table}

Table \ref{table_counters_altogether} summarizes the performance of five weighted model counters (\cachet, \addmc, \minictd, \ctd, and \df) as well as two \textdef{virtual best solvers (VBS)}.
For each benchmark, the solving time of \vbs{} is the shortest solving time among all five actual solvers.
Similarly, the time of \vbss{} is the shortest time among four actual solvers, excluding \addmc.
Note that \addmc{} uniquely solves \benchmarkCountUnique{} benchmarks (that are solved by no other tool).
Additionally, \addmc{} is the fastest solver on \benchmarkCountFastestNonUnique{} other benchmarks.
So \addmc{} improves the solving time of \vbs{} on \benchmarkCountFastest{} benchmarks in total.

% We further observe that \addmc{} performs well on the \classBayes{} benchmark class (solving 1082 of \benchmarkCountBayes) but poorly on the \classPseudoweighted{} benchmarks (solving 322 of \benchmarkCountPseudoweighted).
% We doubt that this discrepancy is due to the different application domains.
% We analyze both benchmark classes and find several \classPseudoweighted{} benchmarks which have many more variables and CNF clauses than the \classBayes{} benchmarks. % This is a very loose statement. Can we say anything stronger? Or don't say this.

\begin{figure}
  \includegraphics[width=\figwidth]{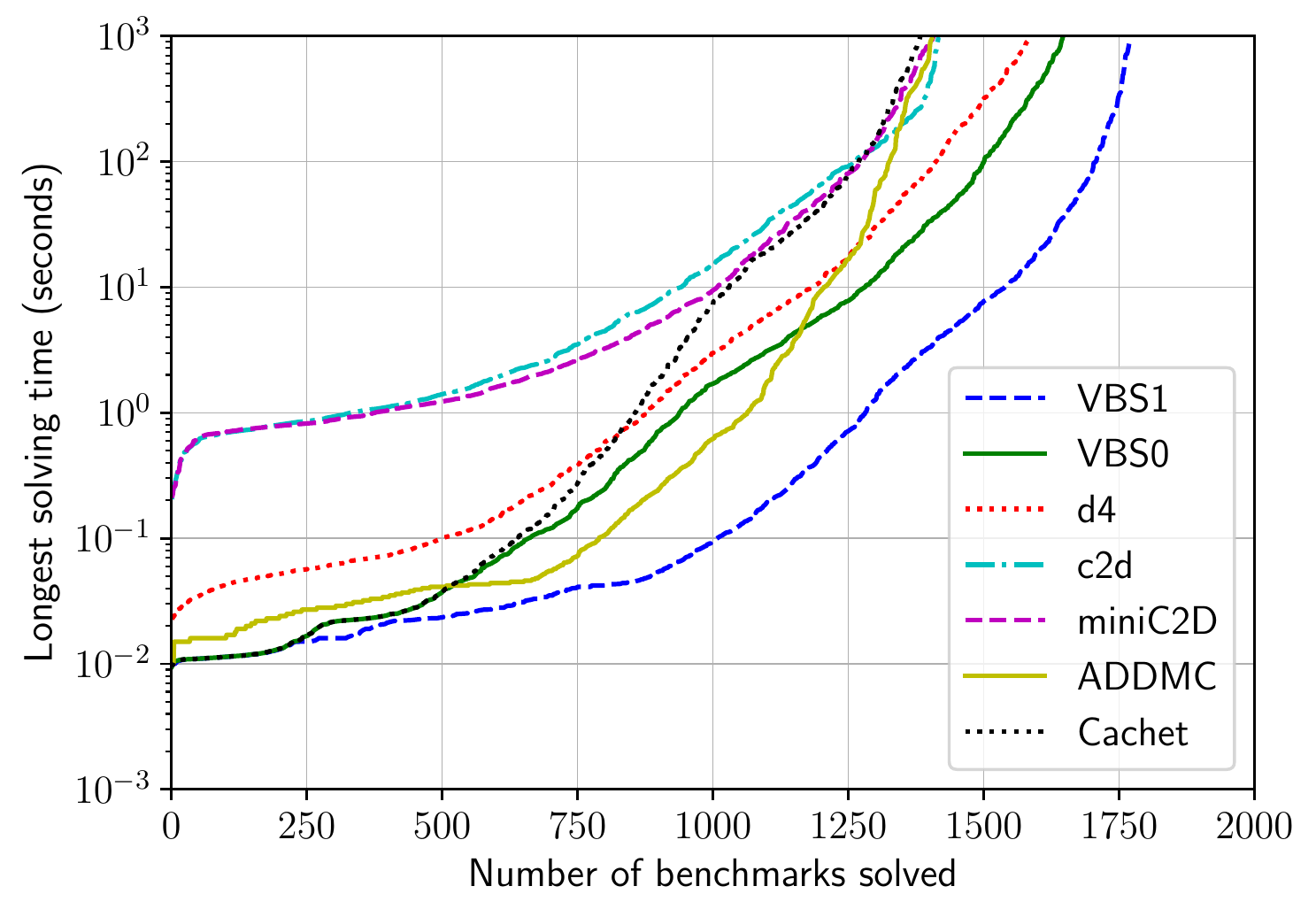}
  \caption{A cactus plot of the numbers of benchmarks solved by five weighted model counters and two virtual best solvers (\vbs{},  with \addmc, and \vbss{}, without \addmc).}
\label{fig_counters}
\end{figure}

See Figure \ref{fig_counters} for a more detailed analysis of the runtime of all solvers.
Evidently, \vbs{} (with \addmc) performs significantly better than \vbss{} (without \addmc).
We conclude that \addmc{} is a useful addition to the portfolio of weighted model counters.

%%%%%%%%%%%%%%%%%%%%%%%%%%%%%%%%%%%%%%%%%%%%%%%%%%%%%%%%%%%%%%%%%%%%%%%%%%%%%%%%
\subsubsection{Predicting \addmc{} Performance}

Generally, \addmc{} can solve a benchmark quickly if all intermediate ADDs constructed during the model counting process are small.
An ADD is small when it achieves high compression under a good diagram variable order; predicting this a priori is difficult and is an area of active research.
However, an ADD also tends to be small if it has few variables, which occurs when an \addmc{} heuristic configuration results in many opportunities for early projection. Moreover, the number of variables that occur in each ADD produced by Algorithm \ref{algo_lwmc_cnf} can be computed much faster than computing the full model count (since we do not need to actually construct the ADDs).

Formally, fix an \addmc{} heuristic configuration.
For a given benchmark, define the \textdef{maximum ADD variable count (MAVC)} to be the largest number of variables across all ADDs that would be constructed when running Algorithm \ref{algo_lwmc_cnf}.
Using the heuristic configuration of Experiment 2 (\besto), we were able to compute the MAVCs of \benchmarkCountMavc{} benchmarks (of \benchmarkCountAltogether{} in total).
We were unable to compute the MAVCs of the remaining 8 benchmarks within \timeoutMavc{} seconds due to the large number of variables and clauses; these benchmarks were also not solved by \addmc{}.

\begin{figure}
  \includegraphics[width=\figwidth]{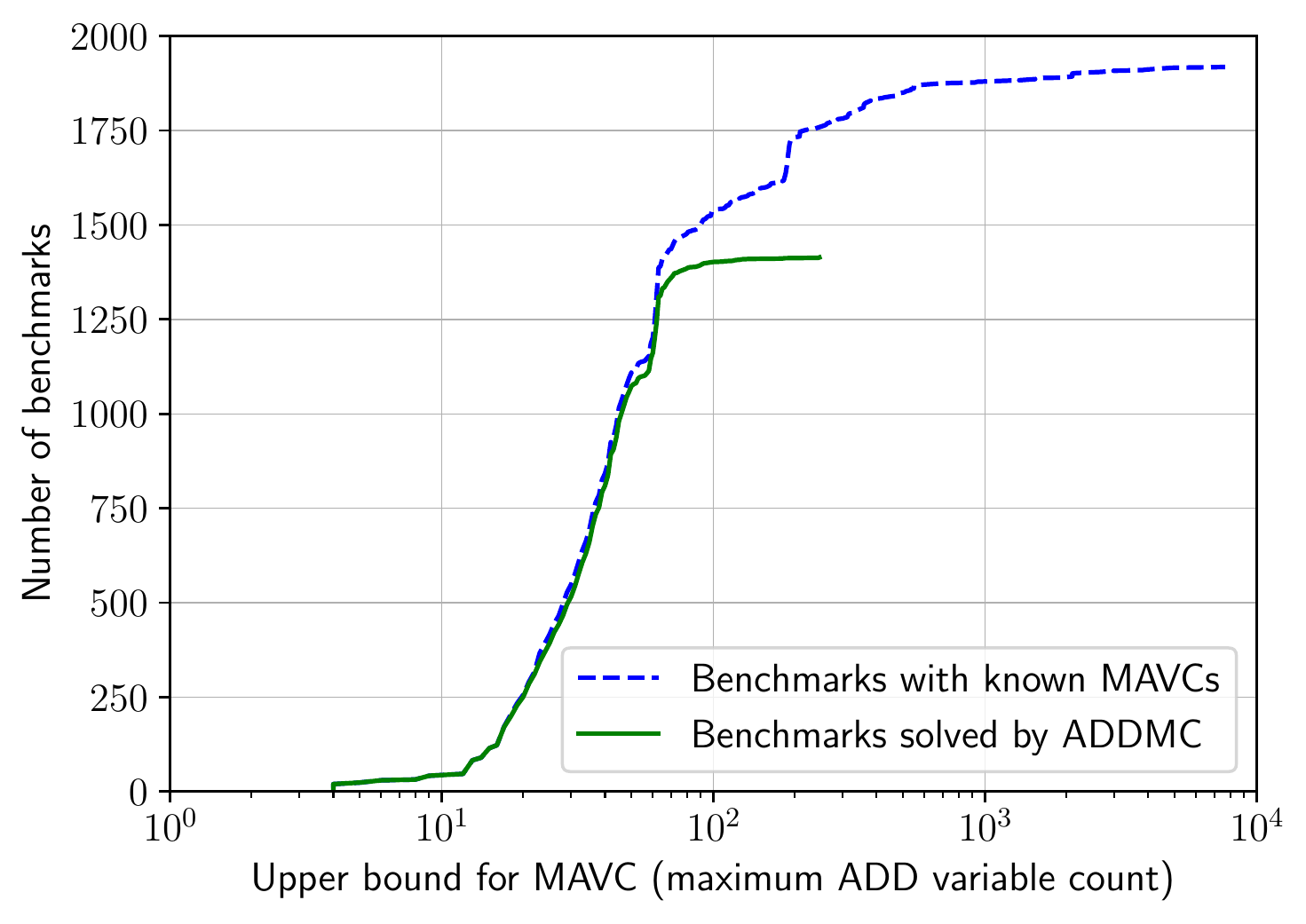}
  \caption{
    A cactus plot of the number of benchmarks, in total and solved by \addmc{}, for various upper bounds on the MAVC.
    % MAVCs of all \benchmarkCountMavc{} benchmarks range from \minMavc{} to \maxMavc{}. % Jeff: Removed this since we were not able to compute MAVCs for all benchmarks.
    The MAVCs of the \benchmarkCountSolved{} benchmarks solved by \addmc{} within \timeoutCounters{} seconds range from \minMavc{} to \maxSolvedMavc{}. % typo in AAAI camera-ready version: MVACs should be MAVCs
  }
\label{fig_mavc_benchmarks}
\end{figure}

Figure \ref{fig_mavc_benchmarks} shows the number of benchmarks solved by \addmc{} in Experiment 2 for various upper bounds on the MAVC.
Generally, \addmc{} performed well on benchmarks with low MAVCs.
In particular, \addmc{} solved most benchmarks (1345 of 1425) with MAVCs less than 70 but solved solved few benchmarks (12 of 379) with MAVCs greater than 100.

\begin{figure}
  \includegraphics[width=\figwidth]{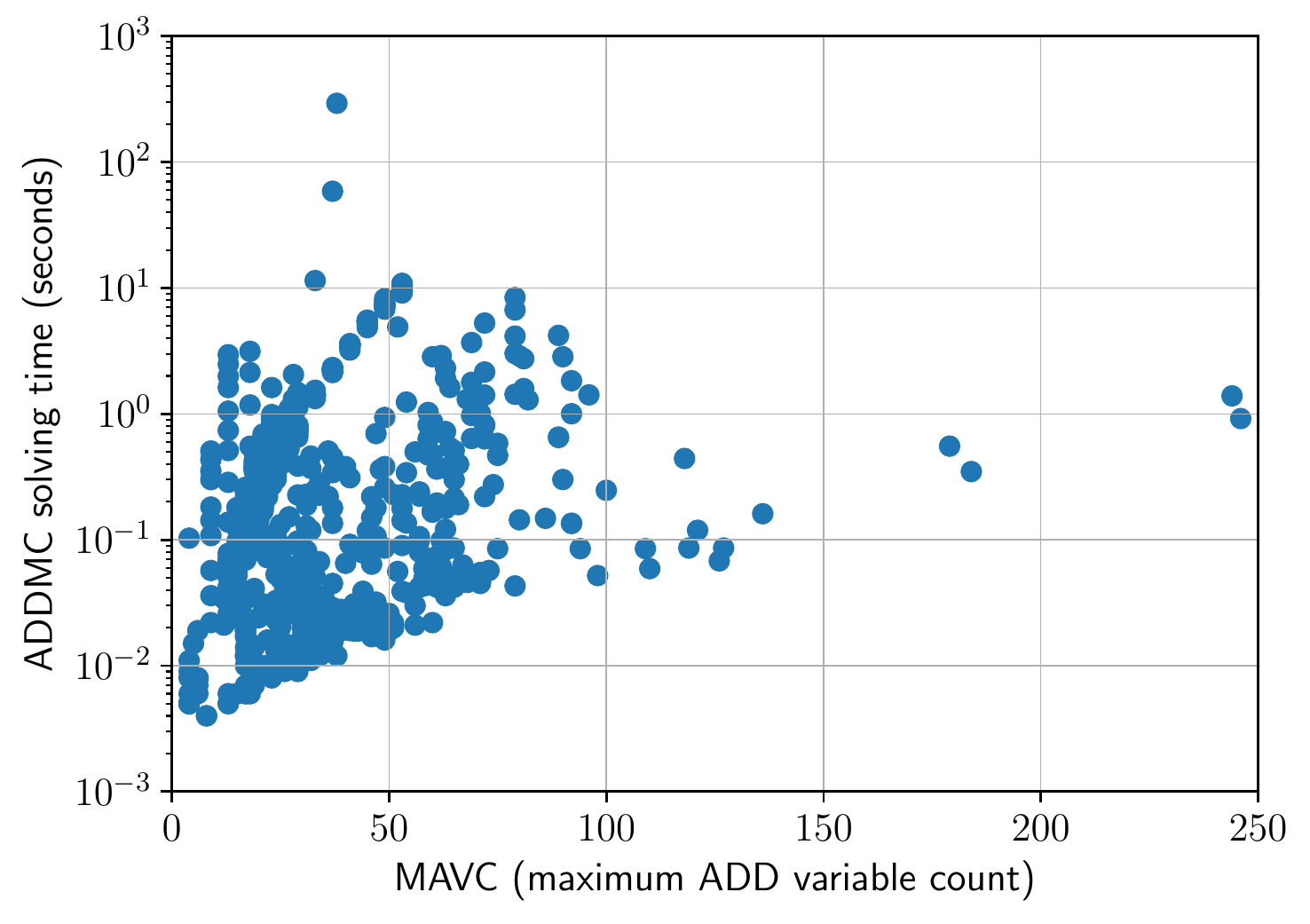}
  \caption{A scatter plot of the solving time of \addmc{} against the MAVC for each of the \benchmarkCountSolved{} benchmarks solved by \addmc{} within \timeoutCounters{} seconds.}
\label{fig_mavc_times}
\end{figure}

Figure \ref{fig_mavc_times} shows the runtime of \addmc{} on the \benchmarkCountSolved{} benchmarks \addmc{} was able to solve in Experiment 2.
In general, \addmc{} was slower on benchmarks with higher MAVCs.

From these two observations, we conclude that the MAVC of a benchmark (under a particular heuristic configuration) is a good predictor of \addmc{} performance.

%%%%%%%%%%%%%%%%%%%%%%%%%%%%%%%%%%%%%%%%%%%%%%%%%%%%%%%%%%%%%%%%%%%%%%%%%%%%%%%%
%% file end 5-experiments.tex
%%%%%%%%%%%%%%%%%%%%%%%%%%%%%%%%%%%%%%%%%%%%%%%%%%%%%%%%%%%%%%%%%%%%%%%%%%%%%%%%

%%%%%%%%%%%%%%%%%%%%%%%%%%%%%%%%%%%%%%%%%%%%%%%%%%%%%%%%%%%%%%%%%%%%%%%%%%%%%%%%
%% file start 6-discussion.tex
%%%%%%%%%%%%%%%%%%%%%%%%%%%%%%%%%%%%%%%%%%%%%%%%%%%%%%%%%%%%%%%%%%%%%%%%%%%%%%%%

\section{Discussion}
\label{sec_discussion}

In this work, we developed a dynamic-programming framework for weighted model counting that captures both bucket elimination and Bouquet's Method.
We implemented this algorithm in \addmc, a new weighted model counter.
We used \addmc{} to compare bucket elimination and Bouquet's Method across a variety of variable order heuristics on \benchmarkCountAltogether{} standard model counting benchmarks and concluded that Bouquet's Method is competitive with bucket elimination.

Moreover, we demonstrated that \addmc{} is competitive with existing state-of-the-art weighted model counters on these 1914 benchmarks.
In particular, adding \addmc{} allows the virtual best solver to solve 124 more benchmarks.
Thus \addmc{} is valuable as part of a portfolio of solvers, and ADD-based approaches to model counting in general are promising and deserve further study.
One direction for future work is to investigate how benchmark properties (\eg, treewidth) correlate with the performance of ADD-based approaches to model counting.
Predicting the performance of tools on CNF benchmarks is an active area of research in the SAT solving community \cite{xu2008satzilla}.

Bucket elimination has been well-studied theoretically, with close connections to treewidth and tree decompositions \cite{dechter99,CD07}. % direct quotes in comments below:
% The induced-width describes the largest cluster in a tree embedding of that graph (also known as treewidth). The complexity of bucket-elimination is time and space exponential in the induced width of the problem's interaction graph \cite{dechter99}.
% Variable elimination [Zhang and Poole, 1996; Dechter, 1996] (VE) is a well–known algorithm for answering probabilistic queries with respect to a Bayesian network. The algorithm runs in time and space exponential in the treewidth of the network. \cite{CD07}
On the other hand, Bouquet's Method is much less well-known.
% Previous work \cite{PV04} has also observed that Bouquet's Method can outperform bucket elimination using compact representations.
Another direction for future work is to develop a theoretical framework to explain the relative performance between bucket elimination and Bouquet's Method.

In this work, we focused on ADDs implemented in the ADD package \cudd.
There are other ADD packages that may be fruitful to explore in the future.
For example, \tool{Sylvan} \cite{van2015sylvan} supports multicore operations on ADDs, which would allow us to investigate the impact of parallelism on our techniques.
Moreover, \tool{Sylvan} supports arbitrary-precision arithmetic.

Several other compact representations have been used in dynamic-programming frameworks for related problems.
For example, AND/OR Multi-Valued Decision Diagrams \cite{mateescu2008and}, Probabilistic Sentential Decision Diagrams \cite{shen2016tractable}, and Probabilistic Decision Graphs \cite{jaeger2004probabilistic} have all been used for Bayesian inference.
Moreover, weighted decision diagrams have been used for optimization \cite{hooker2013decision}, and Affine Algebraic Decision Diagrams have been used for planning \cite{sanner2005affine}.
It would be interesting to see if these compact representations also improve dynamic-programming frameworks for model counting.

%%%%%%%%%%%%%%%%%%%%%%%%%%%%%%%%%%%%%%%%%%%%%%%%%%%%%%%%%%%%%%%%%%%%%%%%%%%%%%%%
%% file end 6-discussion.tex
%%%%%%%%%%%%%%%%%%%%%%%%%%%%%%%%%%%%%%%%%%%%%%%%%%%%%%%%%%%%%%%%%%%%%%%%%%%%%%%%

%%%%%%%%%%%%%%%%%%%%%%%%%%%%%%%%%%%%%%%%%%%%%%%%%%%%%%%%%%%%%%%%%%%%%%%%%%%%%%%%

\section*{Acknowledgments}
The authors would like to thank Dror Fried, Aditya A. Shrotri, and Lucas M. Tabajara for useful comments.
This work was supported in part by the NSF (grants CNS-1338099, % NOTS
IIS-1527668, CCF-1704883, IIS-1830549, % Moshe
and DMS-1547433), by the Ken Kennedy Institute Computer Science \& Engineering Enhancement Fellowship (funded by the Rice Oil \& Gas HPC Conference), by the Ken Kennedy Institute for Information Technology 2017/2018 Cray Graduate Fellowship, % Jeff
and by Rice University.

%%%%%%%%%%%%%%%%%%%%%%%%%%%%%%%%%%%%%%%%%%%%%%%%%%%%%%%%%%%%%%%%%%%%%%%%%%%%%%%%

\bibliographystyle{aaai.bst}
\bibliography{ADDMC.bib}

%%%%%%%%%%%%%%%%%%%%%%%%%%%%%%%%%%%%%%%%%%%%%%%%%%%%%%%%%%%%%%%%%%%%%%%%%%%%%%%%

\clearpage

\appendix

%%%%%%%%%%%%%%%%%%%%%%%%%%%%%%%%%%%%%%%%%%%%%%%%%%%%%%%%%%%%%%%%%%%%%%%%%%%%%%%%
%% file start supplement.tex
%%%%%%%%%%%%%%%%%%%%%%%%%%%%%%%%%%%%%%%%%%%%%%%%%%%%%%%%%%%%%%%%%%%%%%%%%%%%%%%%

%%%%%%%%%%%%%%%%%%%%%%%%%%%%%%%%%%%%%%%%%%%%%%%%%%%%%%%%%%%%%%%%%%%%%%%%%%%%%%%%
\section{ADD Graphical Example}
%%%%%%%%%%%%%%%%%%%%%%%%%%%%%%%%%%%%%%%%%%%%%%%%%%%%%%%%%%%%%%%%%%%%%%%%%%%%%%%%

Figure \ref{fig_add} is a graphical example of an ADD.

\begin{figure}%[H]
  \centering
  \includegraphics[
    trim={2in .5in 0in 1in} % <left> <lower> <right> <upper>
  ]
  {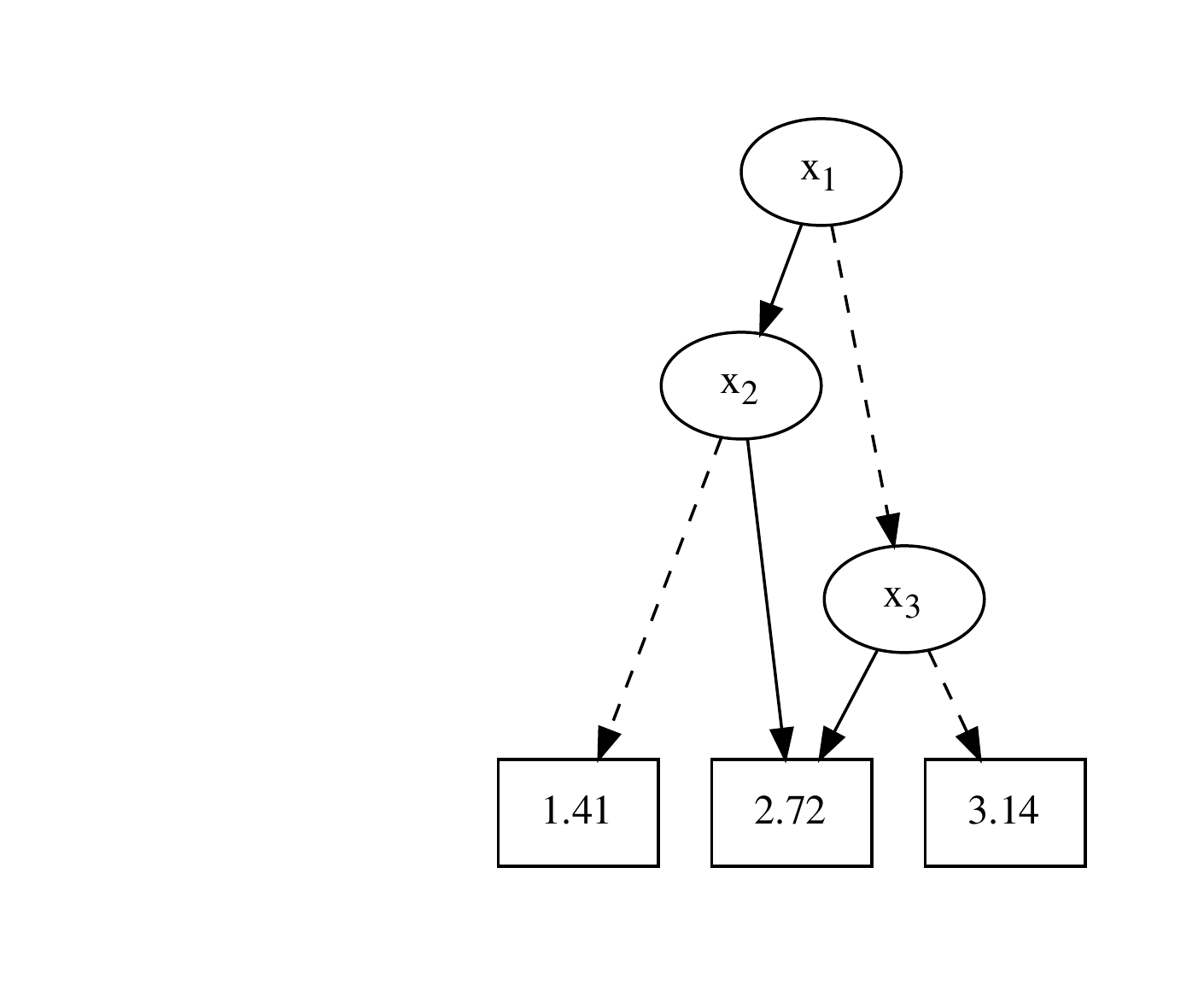}
  \caption{
    The graph $G$ of an ADD with variable set $X = \set{x_1, x_2, x_3}$, carrier set $S = \R$, and diagram variable order $\pi(x_i) = i$ for $i = 1, 2, 3$.
    If a directed edge from an oval node is solid (respectively dashed), then corresponding Boolean variable is assigned 1 (respectively 0).
  }
\label{fig_add}
\end{figure}

%%%%%%%%%%%%%%%%%%%%%%%%%%%%%%%%%%%%%%%%%%%%%%%%%%%%%%%%%%%%%%%%%%%%%%%%%%%%%%%%
\section{Proofs}
%%%%%%%%%%%%%%%%%%%%%%%%%%%%%%%%%%%%%%%%%%%%%%%%%%%%%%%%%%%%%%%%%%%%%%%%%%%%%%%%

%%%%%%%%%%%%%%%%%%%%%%%%%%%%%%%%%%%%%%%%%%%%%%%%%%%%%%%%%%%%%%%%%%%%%%%%%%%%%%%%
\subsection{Proof of Theorem \ref{theorem_wmc}}

\begin{proof}
  Assume the variables in $X$ are $x_1, x_2, \ldots, x_n$.
  Now, for an arbitrary function $A : 2^X \to \R$, we have:
    \begin{align*}
      \sum_{\tau \in 2^X} A(\tau)
      & = \sum_{\tau \in 2^{X \setminus \set{x_n}}} \pars{A(\tau) + A(\tau \cup \set{x_n})}
        \tag{regrouping terms} \\
      & = \sum_{\tau \in 2^{X \setminus \set{x_n}}} ((\exists_{x_n} A)(\tau)) \tag{Definition \ref{def_add_project}} \\
    \end{align*}
  Similarly, projecting all variables in $X$:
    \begin{align*}
      \sum_{\tau \in 2^X} A(\tau)
      & = \sum_{\tau \in 2^{X \setminus \set{x_n}}} (\exists_{x_n} A)(\tau) \\
      & = \sum_{\tau \in 2^{X \setminus \set{x_{n - 1}, x_n}}} (\exists_{x_{n - 1}} \exists_{x_n} A)(\tau) \\
      & \vdots \\
      & = \sum_{\tau \in 2^\sete} (\exists_{x_1} \ldots \exists_{x_{n - 1}} \exists_{x_n} A)(\tau) \\
      & = (\exists_{x_1} \ldots \exists_{x_{n - 1}} \exists_{x_n} A)(\sete) \\
      & = (\exists_X A)(\sete) \\
    \end{align*}
  When $A$ is the specific function $\varphi \mult W : 2^X \to \R$, we have:
    \begin{align*}
      \sum_{\tau \in 2^X} (\varphi \mult W)(\tau)
      & = (\exists_X (\varphi \mult W))(\sete)
    \end{align*}
  Finally:
    \begin{align*}
      (\exists_X (\varphi \mult W))(\sete)
      & = \sum_{\tau \in 2^X} (\varphi \mult W)(\tau) \tag{as above} \\
      & = \sum_{\tau \in 2^X} \varphi(\tau) \mult W(\tau) \tag{Definition \ref{def_add_mult}} \\
      & = W(\varphi) \tag{Definition \ref{def_wmc}}
    \end{align*}
\end{proof}

%%%%%%%%%%%%%%%%%%%%%%%%%%%%%%%%%%%%%%%%%%%%%%%%%%%%%%%%%%%%%%%%%%%%%%%%%%%%%%%%
\subsection{Proof of Theorem \ref{theorem_early_project}}

\begin{proof}
  For every $\tau \in 2^{(X \cup Y) \setminus \set{x}}$, we have:
    \begin{align*}
      (\exists_x &(A \mult B)) (\tau) \\
      & = (A \mult B)(\tau) + (A \mult B)(\tau \cup \set{x})
        \tag{Definition \ref{def_add_project}} \\
      & = A(\tau \cap X) \mult B(\tau \cap Y)+ A((\tau \cup \set{x}) \cap X) \mult B((\tau \cup \set{x}) \cap Y)
        \tag{Definition \ref{def_add_mult}} \\
      & = A(\tau \cap X) \mult B(\tau \cap Y)+ A((\tau \cup \set{x}) \cap X) \mult B(\tau \cap Y)
        \tag{as $x \notin Y$} \\
      & = A(\tau \cap X) \mult B(\tau \cap Y)+ A(\tau \cap X \cup \set{x}) \mult B(\tau \cap Y)
        \tag{as $x \in X$} \\
      & = (A(\tau \cap X) + A(\tau \cap X \cup \set{x})) \mult B(\tau \cap Y)
        \tag{common factor} \\
      & = (\exists_x A)(\tau \cap X) \mult B(\tau \cap Y)
        \tag{Definition \ref{def_add_project}} \\
      & = (\exists_x A)(\tau \cap (X \setminus \set{x})) \mult B(\tau \cap Y)
        \tag{as $x \notin \tau$} \\
      & = ((\exists_x A) \mult B)(\tau)
        \tag{Definition \ref{def_add_mult}}
    \end{align*}
\end{proof}

%%%%%%%%%%%%%%%%%%%%%%%%%%%%%%%%%%%%%%%%%%%%%%%%%%%%%%%%%%%%%%%%%%%%%%%%%%%%%%%%
\subsection{Proof of Theorem \ref{theorem_algo_correct}}

In order to prove Theorem \ref{theorem_algo_correct}, we first state and prove two invariants that hold during the loop at line \ref{line_loop2} of Algorithm \ref{algo_lwmc_cnf}.

First, we prove in the following lemma that the variables in $X_i$ never appear in the clusters $\kappa_j$ for every $i < j$.
\begin{lemma}
  \label{lemma_vars_disjoint}
  Assume the conditions of Theorem \ref{theorem_algo_correct}.
  At every step of the loop at line \ref{line_loop2} of Algorithm \ref{algo_lwmc_cnf} and for every $1 \leq i < j \leq m$, $X_i \cap \vars(\kappa_j) = \sete$.
\end{lemma}
\begin{proof}
  We prove this invariant by induction on the steps of the algorithm.
  The base case (\ie, immediately before line \ref{line_loop2}) follows from the initial construction of $X_i$.

  During the loop, the only potential problem is at line \ref{line_add_to_cluster}.
  In particular, consider an iteration $i < m$ where some ADD $A_i$ is added to $\kappa_j$ (where $j = \chooseCluster(A_i, i)$) and assume that the invariant holds before line \ref{line_add_to_cluster}.
  To prove that the invariant still holds after line \ref{line_add_to_cluster}, consider some $1 \leq s < j$.
  We prove by cases that $X_s \cap \vars(A_i) = \sete$:
  \begin{itemize}
      \item \textbf{Case $s < i$.} By the inductive hypothesis, we have $X_s \cap \vars(\kappa_i) = \sete$.
      Since $\vars(A_i) \subseteq \vars(\kappa_i)$, it follows that $X_s \cap \vars(A_i) = \sete$.
      \item \textbf{Case $s = i$.} All variables in $X_i$ are projected from $A_i$ during the loop at line \ref{line_loop2_inner}.
      Thus $X_i \cap \vars(A_i) = \sete$.
      \item \textbf{Case $s > i$.} Since $s < j$, it follows from Condition 3 of Theorem \ref{theorem_algo_correct} that $X_s \cap \vars(A_i) = \sete$.
  \end{itemize}
  By the inductive hypothesis, $X_i \cap \vars(B) = \sete$ for all other $B \in \kappa_j$.
  Hence $X_i \cap \vars(\kappa_j) = \sete$.
\end{proof}

Next, we use this invariant prove that the ADDs in $\bigcup_{j \geq i} \kappa_j$ always contain sufficient information to compute the weighted model count at iteration $i$.
\begin{lemma}
  \label{lemma_remaining_clusters}
  Assume the conditions of Theorem \ref{theorem_algo_correct} and let $Y_i = \bigcup_{j \geq i} X_i$.
  At the start of every iteration $i$ of the loop at line \ref{line_loop2} of Algorithm \ref{algo_lwmc_cnf},
  \begin{equation}
  \label{eq_remaining_clusters}
      W(\varphi) = \left( \exists_{Y_i} \pars{\prod_{\substack{j \geq i \\B \in \kappa_j}} B \mult \prod_{x \in Y_i} W_x} \right) ( \sete ).
  \end{equation}
\end{lemma}
\begin{proof}
  We prove this invariant by induction on $i$.

  We first consider iteration $i=1$.
  It follows from Condition 1 of Theorem \ref{theorem_algo_correct} that $\bigcup_{j \geq 1} \Gamma_j = \varphi$.
  Thus $Y_1 = \vars(\varphi) = X$ and moreover $\prod_{j \geq 1} \prod_{B \in \kappa_j} B = \prod_{\gamma \in \varphi} \get{clause-ADD}(\gamma, \pi)$.
  Equation \ref{eq_remaining_clusters} therefore follows directly from Theorem \ref{theorem_wmc}.

  Next, assume that Equation \ref{eq_remaining_clusters} holds at the start of some iteration $i < m$ and consider Equation \ref{eq_remaining_clusters} at the start of iteration $i+1$.
  For convenience, let $\kappa_j$ refer to its value at the start of iteration $i$ and let $\kappa_j'$ refer to the value of $\kappa_j$ at the start of iteration $i+1$ (for all $j \geq i$).

  If $\kappa_i = \sete$, then $\kappa_i$ does not contribute to Equation \ref{eq_remaining_clusters}, so Equation \ref{eq_remaining_clusters} remains unchanged (and thus still holds) at the start of iteration $i+1$.
  If $\kappa_i \neq \sete$, then after lines 10-12 Algorithm \ref{algo_lwmc_cnf} computes $A_i = \exists_{X_i} \left( \prod_{D \in \kappa_i} D \mult \prod_{x \in X_i} W_x \right)$ (using Theorem \ref{theorem_early_project} to rearrange terms).
  By Condition 2 of Theorem \ref{theorem_algo_correct}, $A_i$ is then placed in $\kappa_j'$ for some $i < j \leq m$.
  Therefore, at the start of iteration $i+1$ we have
  \begin{align*}
    & \exists_{Y_{i+1}} \pars{\prod_{\substack{j \geq i+1 \\B \in \kappa_j'}} B \mult \prod_{x \in Y_{i+1}} W_x} \\
    = & \exists_{Y_{i+1}} \pars{A_i \mult \prod_{\substack{j \geq i+1 \\B \in \kappa_j}} B \mult \prod_{x \in Y_{i+1}} W_x}.
  \end{align*}

  Plugging in the value of $A_i$, this is equal to
  $$\exists_{Y_{i+1}} \pars{\left( \exists_{X_i} \prod_{D \in \kappa_i} D \mult \prod_{x \in X_i} W_x \right) \mult \prod_{\substack{j \geq i+1 \\B \in \kappa_j}} B \mult \prod_{x \in Y_{i+1}} W_x}.$$

  Notice $Y_i$ is the disjoint union of $Y_{i+1}$ and $X_i$.
  Thus $X_i \cap \vars(W_x) = \sete$ for all $x \in Y_{i+1}$.
  Moreover, by Lemma \ref{lemma_vars_disjoint} $X_i \cap \vars \left( \prod_{j \geq i+1} \prod_{B \in \kappa_j} B \right) = \sete$.
  It thus follows from Theorem \ref{theorem_early_project} that
  \begin{align*}
    & \exists_{Y_{i+1}} \pars{\prod_{\substack{j \geq i+1 \\B \in \kappa_j'}} B \mult \prod_{x \in Y_{i+1}} W_x} \\
    = & \exists_{Y_{i+1}} \pars{\exists_{X_i} \prod_{\substack{j \geq i \\B \in \kappa_j}} B \mult \prod_{x \in Y_{i}} W_x}.
  \end{align*}

  By the inductive hypothesis, this ADD is exactly $W(\varphi)$ when evaluated at $\sete$.
  It follows that Equation 3 holds at the start of iteration $i+1$ as well.
\end{proof}

Given this second invariant, the proof of Theorem \ref{theorem_algo_correct} is straightforward.
\begin{proof}
  By Lemma \ref{lemma_remaining_clusters}, at the start of iteration $m$ we know that
  \begin{align*}
    W(\varphi) = \left( \exists_{X_m} \pars{\prod_{B \in \kappa_m} B \mult \prod_{x \in Y_m} W_x} \right)(\sete).
  \end{align*}
  Since
  \begin{align*}
    A_m = \exists_{X_m} \pars{\prod_{B \in \kappa_m} B \mult \prod_{x \in Y_m} W_x},
      \tag{using Theorem \ref{theorem_early_project} to rearrange terms}
  \end{align*}
  it follows that line \ref{line_return} returns $W(\varphi)$.
\end{proof}

%%%%%%%%%%%%%%%%%%%%%%%%%%%%%%%%%%%%%%%%%%%%%%%%%%%%%%%%%%%%%%%%%%%%%%%%%%%%%%%%
%% file end supplement.tex
%%%%%%%%%%%%%%%%%%%%%%%%%%%%%%%%%%%%%%%%%%%%%%%%%%%%%%%%%%%%%%%%%%%%%%%%%%%%%%%%

%%%%%%%%%%%%%%%%%%%%%%%%%%%%%%%%%%%%%%%%%%%%%%%%%%%%%%%%%%%%%%%%%%%%%%%%%%%%%%%%

\end{document}